\documentclass[11pt]{article}
\usepackage{amsmath,amsthm,amsfonts}
\usepackage{fullpage,microtype}
\usepackage[charter]{mathdesign}
\usepackage{graphicx}

\newcommand{\C}{\mathbb{C}}

\begin{document}

\theoremstyle{definition}
\newtheorem{theorem}{Theorem}
\newtheorem{definition}[theorem]{Definition}
\newtheorem{lemma}[theorem]{Lemma}
\newtheorem{claim}[theorem]{Claim}
\newtheorem{proposition}[theorem]{Proposition}
\newtheorem{exercise}[theorem]{Exercise}
\newtheorem{corollary}[theorem]{Corollary}
\newtheorem{conjecture}[theorem]{Conjecture}

\newcommand{\Z}{{\mathbb{Z}}}
\newcommand{\Exp}{\mathop\mathbb{E}}
\newcommand{\Var}{\mathop\mathrm{Var}}
\newcommand{\Cov}{\mathop\mathrm{Cov}}
\newcommand{\abs}[1]{\left| #1 \right|}
\newcommand{\wf}{\widehat{f}}
\newcommand{\wg}{\widehat{g}}
\newcommand{\din}{d^{\rm in}}
\newcommand{\dout}{d^{\rm out}}
\newcommand{\e}{{\rm e}}
\newcommand{\eps}{\varepsilon}
\newcommand{\inner}[2]{\left\langle #1 , #2 \right\rangle}
\newcommand{\bra}[1]{\left\langle #1 \right|}
\newcommand{\ket}[1]{\left| #1 \right\rangle}
\newcommand{\bbra}[1]{\bigl\langle #1 \bigr|}
\newcommand{\bket}[1]{\bigl| #1 \bigr\rangle}
\newcommand{\norm}[1]{\left| #1 \right|}
\newcommand{\vsym}{V_{\rm sym}}
\newcommand{\pisym}{\Pi_{\rm sym}}
\newcommand{\jun}{j_{\rm undirected}}
\newcommand{\eqdef}{\triangleq}
\newcommand{\ddz}{\frac{\mathrm{d}}{\mathrm{d}z}}
\newcommand{\Ferm}{\mathrm{Ferm}}
\newcommand{\Imm}{\mathrm{Imm}}
\newcommand{\poly}{\mathrm{poly}}
\newcommand{\oplusP}{\oplus \mathrm{P}}
\newcommand{\NP}{\mathrm{NP}}
\newcommand{\BPP}{\mathrm{BPP}}
\newcommand{\RP}{\mathrm{RP}}

\title{The complexity of the fermionant,\\ and immanants of constant width}
\author{Stephan Mertens \\ Institute for Theoretical Physics \\ Otto von Guericke University Magdeburg \\ and the Santa Fe Institute
\and
Cristopher Moore \\ Computer Science Department \\ University of New Mexico \\ and the Santa Fe Institute}
\maketitle

\begin{abstract}
In the context of statistical physics, Chandrasekharan and Wiese recently introduced the \emph{fermionant} $\Ferm_k$, a determinant-like quantity where each permutation $\pi$ is weighted by $-k$ raised to the number of cycles in $\pi$.  We show that computing $\Ferm_k$ is \#P-hard under Turing reductions for any constant $k > 2$, and is $\oplusP$-hard for $k=2$, even for the adjacency matrices of planar graphs.  As a consequence, unless the polynomial hierarchy collapses, it is impossible to compute the immanant $\Imm_\lambda \,A$ as a function of the Young diagram $\lambda$ in polynomial time, even if the width of $\lambda$ is restricted to be at most $2$.  In particular, if $\Ferm_2$ is in P, or if $\Imm_\lambda$ is in P for all $\lambda$ of width $2$, then $\NP \subseteq \RP$ and there are randomized polynomial-time algorithms for NP-complete problems.
\end{abstract}

\section{Introduction}

Let $A$ be an $n \times n$ matrix.  The \emph{fermionant} of $A$, with parameter $k$, is defined as
\[
\Ferm_k \,A = (-1)^n \sum_{\pi \in S_n} (-k)^{c(\pi)} \prod_{i=1}^n A_{i,\pi(i)} \, ,
\]
where $S_n$ denotes the symmetric group, i.e., the group of permutations of $n$ objects, and $c(\pi)$ denotes the number of cycles in $\pi$.   Chandrasekharan and Wiese~\cite{fermionant} showed that the partition functions of certain fermionic models in statistical physics can be written as fermionants with $k=2$.  This raises the interesting question of whether the fermionant can be computed in polynomial time.

Since $(-1)^{n+c(n)}$ is also the parity of $\pi$, the fermionant for $k=1$ is simply the determinant which can of course be computed in polynomial time. But this appears to be the only value of $k$ for which this is possible:
\begin{theorem}
\label{thm:main}
For any constant $k > 2$, computing $\Ferm_k$ for the adjacency matrix of a planar graph is \#P-hard under Turing reductions.
\end{theorem}
\noindent
Indeed, recent results of Goldberg and Jerrum~\cite{inapprox-tutte} imply that for $k > 2$ the fermionant is hard to compute even approximately.  Unless $\NP \subset \RP$, for $k > 2$ there can be no ``fully polynomial randomized approximation scheme'' (FPRAS) that computes $\Ferm_k$ to within a multiplicative error $1+\eps$ in time polynomial in $n$ and $1/\eps$ with probability at least, say, $3/4$.  

For the physically interesting case $k=2$, we have a slightly weaker result.  Recall that $\oplusP$ is the class of problems of the form ``is $|S|$ odd,'' where membership in $S$ can be tested in polynomial time. 
\begin{theorem}
\label{thm:k2}
Computing $\Ferm_2$ for the adjacency matrix of a planar graph is $\oplusP$-hard.
\end{theorem}
\noindent 
By Valiant and Vazirani~\cite{valiant-vazirani} and more generally Toda's theorem~\cite{toda}, the polynomial hierarchy~\cite{noc} reduces to $\oplusP$ under randomized polynomial-time reductions.  Therefore, if $\Ferm_2$ is in P then the polynomial hierarchy collapses, and in particular $\NP \subseteq \RP$.

Theorems~\ref{thm:main} and~\ref{thm:k2} imply new hardness results for the \emph{immanant}.  Given an irreducible character $\chi_\lambda$ of $S_n$, the immanant $\Imm_\lambda$ of a matrix $A$ is
\[
\Imm_\lambda \,A = \sum_{\pi \in S_n} \chi_\lambda(\pi) \prod_{i=1}^n A_{i,\pi(i)} \, .
\]
In particular, taking $\lambda$ to be the parity or the trivial character yields the determinant and the permanent respectively.  Since the determinant is in P but the permanent is \#P-hard~\cite{valiant}, it makes sense to ask how the complexity of the immanant varies as $\lambda$'s Young diagram ranges from a single row of length $n$ (the trivial representation) to a single column (the parity).  Strengthening earlier results of Hartmann~\cite{hartmann}, B\"urgisser~\cite{burgisser-lower} showed that the immanant is \#P-hard if $\lambda$ is a hook or rectangle of polynomial width, 
\[
\lambda_1 = w \, , \lambda_2 = \cdots = \lambda_{n-w+1} = 1
\quad \text{or} \quad \lambda_1 = \cdots  = \lambda_{n/w} = w \, ,
\]
where $w = \Omega(n^\delta)$ for some $\delta > 0$.  Recently Brylinski and Brylinski~\cite{brylinski} improved these results by showing that the immanant is \#P-hard whenever two successive rows have a polynomial ``overhang,'' i.e., if there is an $i$ such that $\lambda_i - \lambda_{i+1} = \Omega(n^\delta)$ for some $\delta > 0$.  The case where $\lambda$ has large width but small overhang, such as a ``ziggurat'' where $\lambda_i = w-i+1$ and $n=w(w+1)/2$, remains open.  

At the other extreme, Barvinok~\cite{barvinok} and B\"urgisser~\cite{burgisser-upper} showed that the immanant is in P if $\lambda$ is extremely close to the parity, in the sense that the leftmost column contains all but $O(1)$ boxes.  Specifically, \cite{burgisser-upper} gives an algorithm that computes $\Imm_\lambda \,A$ in time $O(s_\lambda d_\lambda n^2 \log n)$, where $s_\lambda$ and $d_\lambda$ are the number of standard and semistandard tableaux of shape $\lambda$ respectively.  If the height of $\lambda$ is $n-c$ where $c=O(1)$, then $s_\lambda$ and $d_\lambda$ are bounded above by ${n \choose c} \sqrt{c!} = O(n^c)$ and ${n \choose c} {n+c-1 \choose c} = O(n^{2c})$ respectively, giving a polynomial-time algorithm.

Any function of the cycle structure of a permutation is a \emph{class function}, i.e., it is invariant under conjugation.  Since any class function is a linear combination of irreducible characters, the fermionant is a linear combination of immanants.  If $k$ is a positive integer, it has nonzero contributions from Young diagrams whose width is $k$ or less:
\begin{equation}
\label{eq:ferm-imm}
\Ferm_k \,A = \sum_\lambda d^{(k)}_\lambda \,\Imm_{\lambda^T} \,A \, ,
\end{equation}
where $\lambda$ ranges over all Young diagrams with depth at most $k$, and $d_\lambda$ denotes the number of semistandard tableaux of shape $\lambda$ and content in $\{1,\ldots,k\}$.  To see this, first note that the class function $k^{c(\pi)}$ is the trace of $\pi$'s action on $(\C^k)^{\otimes n}$ by permutation, i.e., 
\[
\pi (v_1 \otimes v_2 \otimes \cdots v_n) = v_{\pi(1)} \otimes v_{\pi(2)} \otimes \cdots \otimes v_{\pi(n)} \, . 
\]
By Schur duality (see e.g.~\cite{fulton-harris}) the multiplicity of $\lambda$ in $(\C^k)^{\otimes n}$ is $d^{(k)}_\lambda$, so
\[
k^{c(\pi)} = \sum_\lambda d^{(k)}_\lambda \chi_\lambda(\pi) \, . 
\]
We then transform $k^{c(\pi)}$ to $(-1)^n (-k)^{c(\pi)}$ by tensoring each $\lambda$ with the sign representation, flipping it to its transpose $\lambda^T$, which has width at most $k$.

Since there are $O(n^{k-1})=\poly(n)$ Young diagrams of width $k$ or less, for any constant $k$ there is a Turing reduction from the fermionant $\Ferm_k$ to the problem of computing the immanant $\Imm_\lambda$ where $\lambda$ is given as part of the input, and where $\lambda$ has width at most $k$.  Then Theorems~\ref{thm:main} and~\ref{thm:k2} give us the following corollary.
\begin{corollary}
For any constant integer $k$, the problem of computing the immanant $\Imm_\lambda \,A$ as a function of $A$ and $\lambda$ is \#P-hard under Turing reductions if $k \ge 3$, and $\oplusP$-hard if $k=2$, even if $\lambda$ is restricted to Young diagrams with width $k$ or less.  
\end{corollary}
\noindent
In particular, unless the polynomial hierarchy collapses, it cannot be the case that $\Imm_\lambda$ is in P for all $\lambda$ of width $2$.  This is somewhat surprising, since these diagrams are ``close to the parity'' in some sense.  

We conjecture that the fermionant is \#P-hard (as opposed to just $\oplusP$-hard) when $k=2$.  Moreover, we conjecture that the immanant is \#P-hard for any family of Young diagrams of depth $n-n^\delta$, or equivalently, any family with a polynomial number of boxes to the right of the first column:
\begin{conjecture}
\label{conj}
Let $\lambda(n)$ be any family of Young diagrams of depth $n-n^\delta$ for some constant $\delta > 0$.  Then $\Imm_{\lambda(n)}$ is \#P-hard.
\end{conjecture}
\noindent
Roughly speaking, this would imply that the results of~\cite{barvinok,burgisser-upper} showing that $\Imm_\lambda$ is in P if $\lambda$ has depth $n-O(1)$ are tight.

\section{\#P-hardness from circuit partitions and the Tutte polynomial}

\begin{proof}[Proof of Theorem~\ref{thm:main}]  We simply recast known facts about the circuit partition polynomial and the Tutte polynomial.  A \emph{circuit partition} of $G$ is a partition of $G$'s edges into circuits.  Let $r_t$ denote the number of circuit partitions containing $t$ circuits; for instance, $r_1$ is the number of Eulerian circuits.  The \emph{circuit partition polynomial} $j(G;z)$ is the generating function
\begin{equation}
\label{eq:j}
j(G;z) = \sum_{t=1}^\infty r_t z^t \, .
\end{equation}
This polynomial was first studied by Martin~\cite{martin}, with a slightly different parametrization; see also~\cite{arratia,bollobas,bouchet,ellis,jaeger,vergnas79,vergnas88}.

To our knowledge, the fact that $j(G;z)$ is \#P-hard for most values of $z$ first appeared in~\cite{ellis-sarmiento}.  We review the proof here.  Recall that the \emph{Tutte polynomial} of an undirected graph $G=(V,E)$ can be written as a sum over all spanning subgraphs of $G$, i.e., all subsets $S$ of $E$, 
\begin{equation}
\label{eq:tutte}
T(G;x,y) = \sum_{S \subseteq E} (x-1)^{c(S)-c(G)} \,(y-1)^{c(S)+\abs{S}-\abs{V}} \, . 
\end{equation}
Here $c(G)$ denotes the number of connected components in $G$.  Similarly, $c(S)$ denotes the number of connected components in the spanning subgraph $(V,S)$, including isolated vertices.  When $x=y$, 
\begin{equation}
\label{eq:tutte2}
T(G;x,x) = \sum_{S \subseteq E} (x-1)^{c(S)+\ell(S)-c(G)} \, ,
\end{equation}
where $\ell(S) = c(S)+\abs{S}-\abs{V}$ is the total excess of the components of $S$, i.e., the number of edges that would have to be removed to make each one a tree.

\begin{figure}
\begin{center}
\includegraphics[width=\textwidth]{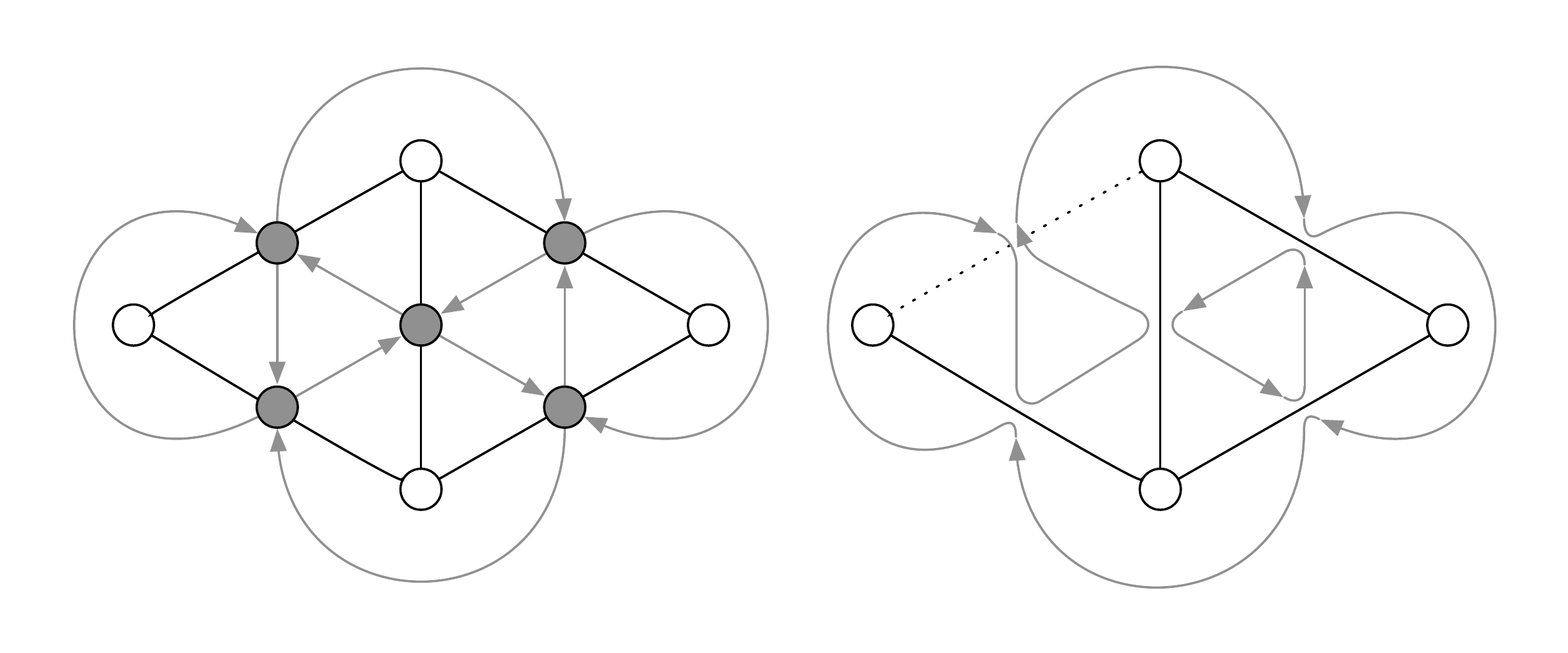}
\end{center}
\caption{Left, a planar graph $G$ (white vertices and black edges) and its medial graph $G_m$ (gray vertices and directed edges).  Right, a subset of the edges of $G$, and the corresponding circuit partition of $G_m$.} 
\label{fig:medial-cycle}
\end{figure}

If $G$ is planar, then we can define its directed medial graph $G_m$.  Each vertex of $G_m$ corresponds to an edge of $G$, edges of $G_m$ correspond to shared vertices in $G$, and we orient the edges of $G_m$ so that they go counterclockwise around the faces of $G$.  Each vertex of $G_m$ has in-degree and out-degree $2$, so $G_m$ is Eulerian.  The following identity is due to Martin~\cite{martin}; see also~\cite{vergnas79}, or~\cite{austin} for a review.
\begin{equation}
\label{eq:martin}
j(G_m;z) = z^{c(G)} \,T(G;z+1,z+1) \, . 
\end{equation}
We illustrate this in Figure~\ref{fig:medial-cycle}.  There is a one-to-one correspondence between subsets $S \subseteq E$ and circuit partitions of $G_m$.  Let $v$ be a vertex of $G_m$ corresponding to an edge $e$ of $G$.  If $e \in S$, the circuit partition ``bounces off $e$,'' connecting each of $v$'s incoming edges to the outgoing edge on the same side of $e$.  If $e \notin S$, the partition crosses over $e$ in each direction.  The number of circuits is then $c(S)+\ell(S)$, in which case~\eqref{eq:tutte2} yields~\eqref{eq:martin}.

Vertigan~\cite{vertigan} proved that the Tutte polynomial for planar graphs is $\#P$-hard under Turing reductions, except on the hyperbolas $(x-1)(y-1) \in \{1, 2\}$ or when $(x,y) \in \{(1,1), (-1,-1), (\omega,\omega^*), (\omega^*,\omega)\}$ where $\omega = \e^{2\pi i /3}$.  It follows that computing $j(G;z)$ is $\#P$-hard unless $z \in \{ 0, -2, \pm 1, \pm \sqrt{2} \}$, and in particular if $z < -2$.

To reduce the circuit partition polynomial to the fermionant, we simply have to turn Eulerian-style circuit partitions, which cover each edge once, into Hamiltonian-style ones, which cover each vertex once.  Given a directed graph $G$, let $G_e$ be its edge graph.  Each vertex of $G_e$ corresponds to an edge $(u,v)$ of $G$, and there is an edge between $(u,v)$ and $(u',v')$ if $v=u'$ so that $(u,v)$ and $(u',v')$ could form a path.  Then each circuit partition of $G$ corresponds to a permutation $\pi$ of the vertices of $G_e$, and 
\[
\Ferm_k \,A_e = j(G;-k) \, ,
\]
where $A_e$ is the adjacency matrix of $G_e$.  In particular, if $G$ is planar, $G_m$ is its medial graph, $G_{m,e}$ is the edge graph of $G_m$ (which is also planar) and $A_{m,e}$ is its adjacency matrix, then
\[
\Ferm_k \,A_{m,e} = (-k)^{c(G)} \,T(G;1-k,1-k) \, ,
\]
and this is \#P-hard for any $k > 2$. 
\end{proof}

As alluded to above, Goldberg and Jerrum~\cite{inapprox-tutte} showed that $T(G;x,y)$ is inapproximable for planar graphs and many values of $x$ and $y$, and in particular in the region $x, y < -1$.  This implies that, unless $\NP \subseteq \RP$, there can be no FPRAS for $\Ferm_k$ for $k > 2$, preventing us from computing the fermionant to within multiplicative error $\eps$ in time polynomial in $n$ and $1/\eps$.  

\section{$\oplusP$-hardness from Hamiltonian circuits}

\begin{proof}[Proof of Theorem~\ref{thm:k2}] 
We will prove Theorem~\ref{thm:k2} by showing that the fermionant $\Ferm_2$ can be used to compute the parity of the number of Hamiltonian circuits in an undirected graph of size $n > 4$.

Let $A$ be the adjacency matrix of an undirected graph $G$ with $n$ vertices. Each permutation $\pi \in S_n$ that gives a non-zero contribution to $\Ferm_2(A)$ corresponds to a Hamiltonian-style circuit partition of $G$; that is, a vertex-disjoint union of undirected cycles that includes every vertex.  A $2$-cycle in $\pi$ corresponds to a circuit that travels back and forth along a single edge.  Any circuit of length greater than $2$ can be oriented in two ways, so a circuit partition with $c_2$ $2$-cycles and $c'$ longer cycles corresponds to $2^{c'}$ permutations.  

Taken together, these permutations contribute $(-1)^n \,(-2)^{c_2} (-4)^{c'}$ to $\Ferm_2 \,A$.  This is a multiple of $8$ unless $n=2$ and $c_2=1$, or $n=4$ and $c_2=2$, or if $c_2=0$ and $c'=1$.  This last case corresponds to a Hamiltonian cycle.  Thus 
\begin{equation}
\label{eq:ferm-h}
\frac{1}{4} \,\Ferm_2 \,A \equiv_2 \#H \, . 
\end{equation}
Valiant~\cite{valiant-parity} showed, using a parsimonious reduction from 3-SAT to \textsc{Hamiltonian Cycle}, that the problem $\oplus$\textsc{Hamiltonian Circuits} of computing $\#H \bmod 2$ is $\oplusP$-complete.  Indeed, he showed this even in the case where $G$ is planar and every vertex has degree $2$ or $3$.  Since~\eqref{eq:ferm-h} gives a reduction from \text{$\oplus$\textsc{Hamiltonian Circuits}} to $\Ferm_2$, this shows that  $\Ferm_2$ is $\oplusP$-hard as well.
\end{proof}


More generally, $\Ferm_k$ is at least as hard as computing the number of Hamiltonian circuits mod $k$.  But for $k \ge 3$, our \#P-hardness result is stronger.

\section{Conclusion}

Is the fermionant \#P-hard for $k=2$?  And are immanants \#P-hard even for some Young diagrams of width $2$, as in Conjecture~\ref{conj}?  The reduction of Theorem~\ref{thm:main} fails in this case, since for $x=y=-1$ the Tutte polynomial is easy.  In particular, if $G=(V,E)$ then
\[
\Ferm_2 \,A_{m,e} = (-2)^{c(G)} \,T(G;-1,-1) = (-2)^{c(G)} \,(-1)^{|E|} \,(-2)^{\dim B} \, .
\]
Here $\dim B$ is the dimension of $G$'s \emph{bicycle space}---the set of functions from $E$ to $\Z_2$ that can be written both as linear combinations of cycles and as linear combinations of stars---which we can determine in polynomial time using linear algebra.

However, there is no reason to think that another reduction couldn't work.  After all, the number of Eulerian circuits of $G_m$ is the number of spanning trees of $G$, which is also in P---but Brightwell and Winkler showed that counting Eulerian circuits is \#P-hard in general~\cite{brightwell-winkler}.

Finally, it is interesting that our proof of \#P-hardness is an indirect reduction from, say, the chromatic polynomial, passing through the Tutte polynomial and its ability to perform polynomial interpolation by decorating the graph.  Thus another open question is whether there is a more direct reduction to the fermionant from the permanent or, say, \textsc{\#Hamiltonian Circuits}.

\paragraph{Acknowledgments.}  This work was supported by the National
Science Foundation grant CCF-1117426, and by the ARO under contract
W911NF-09-1-0483.  We are grateful to Uwe-Jens Wiese for telling us
about the fermionant, and to Alex Russell, Leonard J. Schulman, Leslie Ann Goldberg, and David Gamarnik for helpful discussions.

\end{document}